\documentclass[a4paper,11pt,]{article}
\usepackage[dvipdfmx]{graphicx}
\usepackage{tabularx}
 \usepackage{amssymb}
 \usepackage{amsmath}
\usepackage{multirow}
\usepackage{supertabular}

\newtheorem{theorem}{Theorem}[section]
\newtheorem{lemma}[theorem]{Lemma}

\def\squarebox#1{\hbox to #1{\hfill\vbox to #1{\vfill}}}
\def\qed{\hspace*{\fill}%
        \vbox{\hrule\hbox{\vrule\squarebox{.667em}\vrule}\hrule}\smallskip}
\newenvironment{proof}{\begin{trivlist}
\item[\hspace{\labelsep}{\em\noindent Proof.~}]}{\qed\end{trivlist}}

\begin{document}

\title{
	A tight analysis of Kierstead-Trotter algorithm \\for online unit interval coloring
}

\author{
	Tetsuya Araki, 
	Koji M. Kobayashi
}

\date{}

\maketitle

\begin{abstract}
	Kierstead and Trotter (Congressus Numerantium 33, 1981) proved that 
	their algorithm is an optimal online algorithm for the online interval coloring problem. 
	In this paper, 
	for online unit interval coloring, 
	we show that the number of colors used by the Kierstead-Trotter algorithm is at most $3 \omega(G) - 3$, 
	where $\omega(G)$ is the size of the maximum clique in a given graph $G$, 
	and it is the best possible. 
\end{abstract}

\section{Introduction} \label{Intro}

The online interval coloring problem has been extensively studied for many years. 
Research on this problem is motivated by applications 
such as resource allocation in communication networks. 
Online interval coloring is defined as follows: 
given any interval graph $G = (V, E)$,  
an online algorithm does not have any information on $G$ at first. 
Intervals of $G$ are revealed to the online algorithm one by one over time 
(the length of an interval is one for the unit interval coloring, 
which we study in this paper). 
The algorithm must assign a color to a revealed interval before the next one is revealed such that any two intersecting intervals are not colored by the same color. 
The cost of an algorithm is the number of colors used to color intervals of $G$ 
and the objective of this problem is to minimize the number of used colors. 
Note that the number of colors used by an optimal offline algorithm is the size of the maximum clique in $G$, 
which is denoted as $\omega(G)$.

{\bf Previous Results and Our Results.}~
For any interval graph $G$, 
Kierstead and Trotter~\cite{KT1981} designed an online algorithm and 
proved the number of colors used by the algorithm is at most $3 \omega(G) - 2$. 
Moreover, 
they showed an instance such that the number of colors by any online algorithm is at least $3 \omega(G) - 2$. 
Thus, their online algorithm is optimal for online interval coloring. 
There are various lengths of intervals in this instance, 
which means that it cannot be applied to the case 
in which the lengths are restricted to one, 
that is, unit interval setting. 
In this paper, for online unit interval coloring, 
we conduct a complete analysis of the performance of their algorithm. 
Specifically, 
we prove the number of colors used by their algorithm is at most $3 \omega(G) - 3$. 
In addition, 
we present an instance for which the number of colors used by their algorithm is $3 \omega(G) - 3$.
Epstein and Levy~\cite{EL2005} showed that 
the number of colors used by {\sc First-Fit} is at most $2 \omega(G) - 1$. 
That is, 
our results show that 
the Kierstead-Trotter algorithm is not optimal for the unit interval coloring.

{\bf Related Results.}~
In online unit interval coloring, 
the current best upper and lower bounds are 
$2 \omega(G) - 1$ and $3 \omega(G) / 2$, respectively 
by Epstein and Levy~\cite{EL2005}. 
Furthermore, 
Chrobak and \'Slusarek~\cite{CS1988} and Epstein and Levy~\cite{EL2005} showed that 
a lower bound on the number of colors used by {\sc First-Fit} is $2 \omega(G) - 1$. 
Hence, 
the performance of {\sc First-Fit} for the unit interval case is tightly analyzed. 
On the other hand, 
recent research on online interval coloring has focused on evaluating the performance of {\sc First-Fit}~
\cite{K1988,CS1988,KQ1995,NB2008,PRV2011,KST2016}, 
which still remains open. 
The current best upper and lower bounds are $8 \omega(G)$~\cite{NB2008,PRV2011} and $5 \omega(G)$~\cite{KST2016} respectively. 
Variants of online interval coloring with some constraints and generalizations of that have been extensively studied as well 
(see e.g. \cite{AE2002,N2004,EL2005,AFLN2006,ELI2008,EEL2009,PRV2011}). 
In addition, 
the max coloring problem, which was proposed by Pemmaraju et~al.\ \cite{PRV2011}, 
is a generalization of the vertex coloring problem. 
They pointed out that 
this problem of an interval graph is closely related to the interval coloring problem, and 
Epstein and Levy~\cite{EL2012} studied the max coloring of interval graphs in an online setting.

\section{Kierstead-Trotter algorithm} \label{KTalg}

In this section, 
we give the definition of the algorithm by Kierstead and Trotter \cite{KT1981}, 
which we analyze in this paper. 
First, 
we give some definitions to define it. 
Let $v_i$ be the $i$th interval which is revealed to the online algorithm. 
The algorithm gives each interval $v$ the value $\ell(v)$, called the {\em level} of $v$, and 
sets its value just before coloring $v$. 
The algorithm colors $v$ based on its level. 
(Levels are initialized to one at the beginning of the first call of the algorithm.)
Let us define 
$V_{x, y}(i) = \{ v_j \in V \mid j \leq i, \hspace{2mm} x \leq \ell(v_j) \leq y \}$, 
$E_{x, y}(i) = \{ (u, v) \in E \mid u, v \in V_{x, y}(i) \}$ and 
$G_{x, y}(i) = (V_{x, y}(i), E_{x, y}(i))$. 
Let $P_{j}$ denote the set of colors dedicated to the graph $G_{j, j}(i)$. 
That is, $P_{j} \cap P_{j'} = \emptyset$ if $j \ne j'$. 
For any interval subgraph $H \subseteq G$ and any interval $v$ in $H$, 
$\omega(H, v)$ denotes the size of the maximum clique containing $v$. 
\noindent\vspace{-1mm}\rule{\textwidth}{0.5mm} 
\vspace{-3mm}
{\bf Kierstead-Trotter algorithm}\\
\rule{\textwidth}{0.1mm}
	{\bf\boldmath Initialize: }
	For each interval $v$, set $\ell(v) := 1$. \\
	%
	Suppose that the $i$th interval $v_i$ is revealed. \\
	%
	%
	{\bf\boldmath Step 1: } 
	Set $\ell(v_i) := \arg \min \{ j \mid \omega( G_{1, j}(i), v_i ) \leq j \}$. 
	\\
	%
	%
	{\bf\boldmath Step 2: } 
	Color $v_i$ considering only $G_{\ell(v_i), \ell(v_i)}(i)$ using {\sc First-Fit} on the colors of $P_{\ell(v_i)}$. 
	\\
\rule{\textwidth}{0.1mm}
%

\section{Analysis} \label{MB}

First, we show an upper bound on the number of colors used by Kierstead-Trotter algorithm. 
Let $n$ be the total number of given intervals. 
The following lemma was shown in Theorem~5 of \cite{KT1981}. 

\begin{lemma} \label{LMA:up.1}
	{\sc First-Fit} colors $G_{1, 1}(n)$ using at most one color. 
	Also for any $j \geq 2$, 
	{\sc First-Fit} colors $G_{j, j}(n)$ using at most three colors. 
\end{lemma}

\begin{lemma} \label{LMA:up.2}
	{\sc First-Fit} colors $G_{2, 2}(n)$ using at most two colors. 
\end{lemma}
\begin{proof}
	We prove by contradiction that for any interval $v \in V_{2, 2}(n)$, 
	the number of intervals in $V_{2, 2}(n)$ which intersect $v$ is at most one. 
	We assume that $v$ intersects two intervals $v'$ and $v''(\ne v') \in V_{2, 2}(n)$. 
	Since the lengths of $v, v'$, and $v''$ are unique, 
	the left and right endpoints of $v$ are contained in either $v'$ or $v''$. 
	Then, it is clear that 
	the endpoint which $v'$ contains is different from the endpoint which $v''$ contains. 
	That is, 
	the right (left) endpoint of $v$ is included in a clique whose size is two in $G_{1,2}(n)$. 
	Thus, 
 	no interval exists in $V_{1,1}(n)$ which contains the right (left) endpoint of $v$.  
	On the other hand, 
	$v$ intersects an interval in $V_{1,1}(n)$ by the definition of the algorithm because the level of $v$ is two. 
	Then, 
	any unit interval intersecting $v$ certainly contains at least one of the endpoints of $v$, 
	which contradicts the above fact. 
\end{proof}

\begin{theorem}\label{thm:1}
	For any unit interval graph $G$, 
	the number of colors used by the Kierstead-Trotter algorithm is at most $3 \omega(G) - 3$. 
\end{theorem}
\begin{proof}
	By Lemmas~\ref{LMA:up.1} and \ref{LMA:up.2}, 
	{\sc First-Fit} colors $G_{1, 1}(n)$ and $G_{2, 2}(n)$ using at most one color and at most two colors, respectively. 
	Furthermore, for any $j \geq 3$, 
	{\sc First-Fit} colors $G_{j, j}(n)$ using at most three colors. 
	Therefore, 
	$1 + 2 + 3(\omega(G) - 2) = 3 \omega(G) - 3$, 
	which completes the proof. 
\end{proof}
Next, we show a lower bound on the number of colors used by the Kierstead-Trotter algorithm. 

\begin{theorem}\label{thm:2}
	There exists an instance which gives a graph $G$ 
	such that the number of colors used by the Kierstead-Trotter algorithm is $3 \omega(G) - 3$. 
\end{theorem}
\begin{proof}
	An instance which gives graph $G$ is constructed as follows. 
	Let $x \geq 3$ be an integer. 
	For each $i = 1, \ldots, x+2$, 
	a unit interval is revealed whose left endpoint is located at $(i-1)(1-\frac{1}{x})$. 
	Then, 
	for an interval $v$ which is given when $i$ is even, 
	$v$ is included in a clique whose size is two in the range $[ (i-1)(1-\frac{1}{x}), (i-1)(1-\frac{1}{x})+\frac{1}{x} ]$ just after $v$ is given. 
	Thus 
	the level of $v$ is two. 
	Next, 
	for each $a = 1,2, \ldots$, 
	we define $i_a = a$ if $a \ne 3,4$, 
	$i_3 = 4$, and 
	$i_4 = 3$. 
	Then, 
	for each $j = 2, \ldots, x-1$, 
	in the order of $i = i_1, i_2, \ldots, i_{x - j + 3}$, 
	an interval $v'$ whose left endpoint is located at $(i-1)(1-\frac{1}{x}) + \frac{j-1}{x}$ is revealed. 
	Then, 
	$v'$ is included in a clique whose size is $j+1$ in the range $[ i (1-\frac{1}{x}) + \frac{j-2}{x}, i (1-\frac{1}{x})+\frac{j-1}{x} ]$ 
	just after $v'$ is revealed, 
 	which means that the level of $v'$ is $j+1$. 
	In addition, 
	the first, second, and fourth intervals with level $j+1$ have different colors when using {\sc First-Fit}, namely, 
	{\sc First-Fit} uses three colors for level $j+1$. 
	Finally, 
	four intervals whose left endpoints are located at 
	$(x + 1)(1-\frac{1}{x}) + 1 + \frac{1}{x}$, 
	$(x + 1)(1-\frac{1}{x}) + 2 + \frac{2}{x}$, 
	$(x + 1)(1-\frac{1}{x}) + 1 + \frac{2}{x}$ and 
	$(x + 1)(1-\frac{1}{x}) + 2 + \frac{1}{x}$ 
	are revealed one by one in this order. 
	The levels of the latter two intervals are two, and 
	{\sc First-Fit} uses two colors for them. 
	By the above argument, 
	the maximum level of given intervals is $x$, 
	that is, 
	the size of the maximum clique in $G$ is $x$. 
	Therefore, 
	the total number of colors used by the Kierstead-Trotter algorithm is $1 + 2 + 3(x-2) = 3x-3 = 3 \omega(G) - 3$. 
\end{proof}
%



\end{document}